\documentclass{article}[12pt]
\usepackage{amsmath,graphicx,amsthm,amssymb,mathrsfs,booktabs,mdwlist,multirow,colortbl,bm}
\usepackage{algorithmic,algorithm,authblk}
\usepackage[top=1in, bottom=1in, left=1in, right=1in]{geometry}

\theoremstyle{plain}
\newtheorem{theorem}{Theorem}[section]
\newtheorem{proposition}[theorem]{Proposition}
\newtheorem{lemma}[theorem]{Lemma}

\theoremstyle{definition}
\newtheorem{definition}[theorem]{Definition}
\newtheorem{assumption}[theorem]{Assumption}
\newtheorem{example}[theorem]{Example}
\theoremstyle{remark}
\newtheorem{remark}[theorem]{Remark}

\usepackage[textsize=tiny]{todonotes}

\begin{document}
\title{Efficient Distributed Learning in Stochastic Non-cooperative Games\\
	       without Information Exchange}
\author[a,b]{Haidong Li}
\author[a,b]{Anzhi Sheng}
\author[c]{Yijie Peng \thanks{Corresponding author: pengyijie@pku.edu.cn}}
\author[a,b]{Long Wang \thanks{Corresponding author: longwang@pku.edu.cn}}
\affil[a]{Department of Advanced Manufacturing and Robotics, College of Engineering, Peking University, Beijing, China}
\affil[b]{Center for Systems and Control, College of Engineering, Peking University, Beijing, China}
\affil[c]{Guanghua School of Management, Peking University, Beijing, China}
\date{}
\maketitle

\begin{abstract}
In this work, we study stochastic non-cooperative games, where only noisy black-box function evaluations
are available to estimate the cost function for each player. Since each player's cost function depends on both its own decision variables and its rivals' decision variables, local information needs to be exchanged through a center/network in most existing work for seeking the Nash equilibrium. We propose a new stochastic distributed learning algorithm that does not require communications among players. The proposed algorithm uses simultaneous perturbation method to estimate the gradient of each cost function, and uses mirror descent method to search for the Nash equilibrium. We provide asymptotic analysis for the
bias and variance of gradient estimates, and show the proposed algorithm converges to the Nash equilibrium in mean square for the class of strictly monotone games at a rate faster than the existing algorithms. The effectiveness of the proposed method is buttressed in a numerical experiment.
\end{abstract}

\section{Introduction}
\label{sec:intro}
In game theory, the scenarios of non–cooperative game include that individuals autonomously optimize their
selfish objectives, e.g., congestion control~\cite{alpcan2002game}, Cournot oligopoly~\cite{bimpikis2019cournot}, and resource allocation auctions~\cite{marden2013distributed}. The individual cost function of each player depends on not only its own stratege but also its rivals' strategies. The Nash equilibrium is a solution of non–cooperative game~\cite{nash1950equilibrium}.

In general, players do not have complete information of the game and thus cannot compute a Nash equilibrium in an introspective manner~\cite{lu2018distributed}. For learning the Nash equilibrium, players need to take actions in repeated games, and adjust their strategies in response to their rivals’ strategies for seeking a lower cost. In most existing work, a center that has bidirectional communication with all players is required~\cite{belgioioso2021semi} or agents are allowed to exchange information with their neighbors in a network~\cite{pavel2019distributed,lu2020online}. Learning the Nash equilibrium without local information exchange is considered much more challenging and has not been well studied. In this paper, we study how each player can learn the Nash equilibrium based on only its own cost observations.

We consider stochastic non–cooperative games where each player only has access to a noisy but unbiased estimate of its cost function. Such stochastic setting is known as black-box or zeroth-order~\cite{ghadimi2013stochastic,nesterov2017random}. Stochastic noise complicates the problem considerably because each player must estimate the gradient of its cost function from the observed feedbacks. Stochastic gradient estimation may introduce a significant error in the Nash equilibrium learning process. In~\cite{mertikopoulos2019learning} and~\cite{lei2021distributed}, they assume each player can obtain an unbiased stochastic gradient estimate with a bounded variance, which is not available in our black-box setting. In this paper, we propose a black-box stochastic gradient estimation method and study the bias-variance tradeoff of the proposed method.

Specifically, we propose a simultaneous perturbation method for gradient estimation and mirror descent method for searching for the Nash equilibrium of stochastic non-cooperative games. Simultaneous perturbation method proposed by~\cite{spall1992multivariate} is a black-box gradient estimation method using only noisy outputs of function evaluations. This method has the advantage that the number of function evaluations required to estimate gradient is independent of the dimension of decision variables. Simultaneous perturbation method has been widely used in stochastic gradient-based optimization methods~\cite{robbins1951stochastic}. Mirror descent method, a popular class of no-regret policies for repeated games, can be traced back to~\cite{nemirovskij1983problem} and many variants have been developed~\cite{shalev2006convex,nesterov2009primal,xiao2010dual}. For convex cost function of the game, an $O(\sqrt{n})$ regret bound has been established for the mirror descent method in extensive literature of online learning and optimization~\cite{flammarion2017stochastic,shalev2011online,bubeck2012regret}.

The contribution of our work is threefold.
\begin{itemize}
	\item We propose a stochastic distributed learning algorithm that does not require communications among players. In each period, players perturb their strategies following simultaneous perturbation method to learn local function information; then at the end of each period, each player applies the mirror descent method with a stochastic gradient estimate to seek the Nash equilibrium. 
	\item We establish asymptotic property for the bias and variance of gradient estimates. With an appropriate bias-variance tradeoff, simultaneous perturbation method can efficiently learn gradient information without any knowledge of cost function structure. 
	\item We show the proposed algorithm converges in mean square to the Nash equilibrium for the class of strictly monotone games. Specifically, the rate of convergence of our method is between $O(n^{-1/2})$ and $O(n^{-1})$, which is faster than $O(n^{-1/3})$ in~\cite{bravo2018bandit}.
\end{itemize}

The rest of this paper is organized as follows. In Section~\ref{sec:PF}, we formulate the problem and introduce some assumptions. In Section~\ref{sec:SL}, we propose a stochastic distributed learning algorith and establish some theoretical results. Section~\ref{sec:NS} presents numerical simulation results, and Section~\ref{sec:con} concludes the paper.

\section{Problem Formulation}\label{sec:PF}
In this section, we define stochastic non-cooperative games and introduce some basic assumptions.

\subsection{Problem Statement}
We consider a class of stochastic non-cooperative games $\langle\mathcal{N},(X_{i})_{i=1}^{N},(f_{i})_{i=1}^{N}\rangle$ with a finite set of players $\mathcal{N}\triangleq\{1,\ldots,N\}$. Each player $i\in\mathcal{N}$ chooses its decision variables $x_{i}\in\mathbb{R}^{m_{i}}$ from a strategy set $X_{i}\subset\mathbb{R}^{m_{i}}$. We denote by $x\triangleq(x_{1}^\top,\ldots,x_{N}^\top)^\top\in\mathbb{R}^{d}$ the vector consisting of all decision variables, where $d=\sum_{i=1}^{N}m_{i}$. In non-cooperative games, the cost function of each player $i$ depends on its own decision variables $x_{i}$ and its rivals' decision variables $x_{-i}\triangleq(x_{1}^\top,\ldots,x_{i-1}^\top,x_{i+1}^\top,\ldots,x_{N}^\top)^\top$, i.e., $f_{i}(x_{i},x_{-i})$. The game is to solve the following optimization problem:
\begin{gather}\label{game}
	\forall i\in\mathcal{N}:\underset{x_{i}\in X_{i}}{\min}\ f_{i}(x_{i},x_{-i}).
\end{gather}
Then we aim to compute a Nash equilibrium (NE) of the non-cooperative game~\eqref{game}.
\begin{definition}
	A tuple of decision variables $x^*\in\prod_{i}X_{i}$ is a NE if, for each $i\in\mathcal{N}$,
	$$f_{i}(x_{i}^*,x_{-i}^*)\leq \inf\{f_{i}(x_{i},x_{-i}): x_{i}\in X_{i},x_{-i}=x_{-i}^*\}.$$
\end{definition}

In stochastic regimes, each player has access to a noisy but unbiased estimate $F_{i}(x_{i},x_{-i};\xi_{i})$ of its cost function for any chosen $x$, where the random vector $\xi_{i}\in \mathbb{R}^{\omega_{i}}$ represents the randomness in non-cooperative games, and $F_{i}:\mathbb{R}^{d}\times\mathbb{R}^{\omega_{i}}\to\mathbb{R}$ is a scalar-valued function, and $f_{i}(x_{i},x_{-i})=\mathbb{E}[F_{i}(x_{i},x_{-i};\xi_{i})]$. In other words, the structure of each $f_{i}$ is unknown, and only noisy observations $F_{i}(x_{i},x_{-i};\xi_{i})$ are available to evaluate the cost function $f_{i}(x_{i},x_{-i})$ for each player $i$. In addition, neither a center nor a network is considered in our problem. In other words, each player cannot acquire information of other players' decisions and cost functions.

\begin{remark}
	 We should distinguish two types of assumptions on the knowledge of the cost function in the literature. In multi-agent control theory, it is often assumed that the agents know the structure of their cost functions. Each player’s cost function depends on all players’ actions, but players cannot directly observe actions beyond their neighbors. Therefore, consensus algorithms are required for each player to estimate other players’ decisions. The estimate of cost function can be computed by substituting the estimate of all players’ actions. In the field of simulation optimization, it is often assumed that the structure of cost function is unknown, so players cannot obtain the value of their cost functions even if they know each other's action. A simulation model is used to generate unbiased but noisy function evaluations for estimating the value of cost function. In this study, we consider the second type of assumptions.
\end{remark}

\begin{example}\label{example}
(Cournot competition problem). There is a finite set of $N$ firms competing over $m$ markets denoted by $\mathcal{M}\triangleq\{1,\ldots,m\}$. Each firm $i$ supplies market $j\in\mathcal{M}$ with a quantity $x_{i,j}\in[0,C_{i,j}]$ of products (or service), where $x_{i,j}$ is the $j$-th component of $x_{i}$, and $C_{i,j}$ is the production capacity of firm $i$. By the law of supply and demand, the price $p_{j}$ of products sold in market $j$ is a decreasing function of the total supplying amount $\sum_{i=1}^{N}x_{i,j}$, e.g., $p_{j}(x;\zeta_{j})=a_{j}+\zeta_{j}-b_{j}\sum_{i=1}^{N}x_{i,j}$ with positive constants $a_{j},b_{j}>0$ and a zero-mean random disturbance $\zeta_{j}$. The stochastic cost function of firm $i$ is then given by $F_{i}(x_{i},x_{-i};\zeta,\eta_{i})=\sum_{j=1}^{m}(c_{i,j}+\eta_{i,j}-p_{j}(x;\zeta_{j}))x_{i,j}$ with positive constants $c_{i,j}>0$ and zero-mean random variables $\eta_{i,j}$. Firm $i$ aims to minimize its expected
cost while satisfying the finite capacity constraint.
\end{example}
For more examples of non-cooperative games, see~\cite{scutari2010convex} and~\cite{d2015interference}. 

\subsection{Assumptions}
\begin{assumption}\label{Assump:1}
	For each agent $i\in\mathcal{N}$,\\
	(a) the strategy set $X_{i}$ is closed, convex and compact;\\
	(b) the cost function $f_{i}(x_{i},x_{-i})$ is convex in $x_{i}$ for all $x_{-i}\in\prod_{j\neq i}X_{j}$;\\
	(c) $f_{i}(x_{i},x_{-i})$ is third-order continuously differentiable in $x$;\\
	(d) $\text{Var}[F_{i}(x_{i},x_{-i};\xi_{i})]$ for any $x\in\prod_{i}X_{i}$ is bounded.
\end{assumption}

We define $$\phi(x)\triangleq(\nabla_{x_{1}}f_{1}(x_{1},x_{-1})^\top,\ldots,\nabla_{x_{N}}f_{N}(x_{N},x_{-N})^\top)^{\top}.$$ Under Assumption~\ref{Assump:1}, by~\cite{facchinei2007finite}, $x^*\in\prod_{i}X_{i}$ is a NE if and only if for all $x\in\prod_{i}X_{i}$,
\begin{gather}
	(x-x^{*})^\top\phi(x^{*})\geq0.
\end{gather}
In addition, the existence of NE follows immediately by~\cite{facchinei2007finite}.
\begin{assumption}\label{Assump:2}
	The non-cooperative game is $\beta$-strongly monotone, i.e.,
	$$(\phi(x)-\phi(x'))^\top(x-x')\geq \frac{\beta}{2}\|x-x'\|^{2}$$ for any $x,x'\in \prod_{i}X_{i}$.
\end{assumption}
In this paper, all norms $\|\cdot\|$ are $L^2$ norms. 

\section{Stochastic Learning}\label{sec:SL}
In this section, we propose a stochastic distributed learning algorithm to learn local function information in each period and seek the Nash equilibrium at the end of each period. 

\subsection{Simultaneous Perturbation}
In the black-box setting, a common method for estimating a gradient is to use finite differences. High-dimensional finite-difference scheme perturbs the value of each component of $x$ separately while keeping the other components fixed. Suppose that $f:\mathbb{R}^{d}\to\mathbb{R}$ is a function of interest, where we have access to an unbiased estimate $\widehat{f}(x)$ for any chosen $x\in\mathbb{R}^{d}$. The $i$-th component of central finite-difference gradient estimator is given by
$$\frac{\widehat{f}(x+he_{i})-\widehat{f}(x-he_{i})}{2h},$$
where $h>0$ is a scaling parameter and $e_{i}$ denotes the unit vector in the $i$-th direction. The central finite-difference gradient estimator requires at least $2d$ function evaluations per gradient estimate.

Simultaneous perturbation method uses the same perturbation vector for estimating each component of gradient, which offers the advantage that the number of function evaluations required to estimate gradient is independent of the dimension of decision variables. Specifically, to estimate the gradient $\nabla f(x)$, simultaneous perturbation method gives us the following estimator:
\begin{gather}
	\widehat{\nabla}f(x)=\frac{1}{\ell}\sum_{j=1}^{\ell}\frac{\widehat{f}(x+h\Delta_{j})-\widehat{f}(x-h\Delta_{j})}{2h}\psi(\Delta_{j}),\label{eq:SP}
\end{gather}
where $\ell$ is the number of perturbing pairs and is independent of $d$, each $\Delta_{j}=(\delta_{1,j},\ldots,\delta_{d,j})^\top$ is a perturbation vector, and $\psi(\Delta_{j})=(1/\delta_{1,j},\ldots,1/\delta_{d,j})^\top$. As required in simultaneous perturbation, $\delta_{i,j},\ i=1,\ldots,d$ are all mutually independent with zero-mean, bounded second moments, and $\mathbb{E}(|\delta_{i,j}|^{-1})$ uniformly bounded, and all $\Delta_{j}$'s are independent and identically distributed. A common choice for $\delta_{i,j}$ is the Rademacher distribution, i.e., Bernoulli $\pm1$ with probability $1/2$. Given each $\Delta_{j}$, we generate $\widehat{f}(x+h\Delta_{j})$ and $\widehat{f}(x-h\Delta_{j})$ independently, and $\widehat{f}(x+h\Delta_{j})$'s and $\widehat{f}(x-h\Delta_{j})$'s are also independent across different $\Delta_{j}$'s. 

We define observation variance $\sigma_{f}^{2}(x)\triangleq \text{Var}[\widehat{f}(x)]>0$. The following theorem provides the asymptotic property of the bias and variance of gradient estimate $\widehat{\nabla}f(x)$.
\begin{theorem}\label{thm:mse}
	Suppose $f(x)$ is third-order continuously differentiable, and $\sigma_{f}^{2}(x)$ for any $x\in\prod_{i}X_{i}$ is bounded. Then as $h\to 0$ and $\ell\to\infty$, the order of bias of $\widehat{\nabla}f(x)$ satisfies $$\mathbb{E}[\widehat{\nabla}f(x)]-\nabla f(x)=O(h^2),$$ the order of variance of $\widehat{\nabla}f(x)$ satisfies $$\mathbb{E}\|\widehat{\nabla}f(x)-\mathbb{E}[\widehat{\nabla}f(x)]\|^{2}=O\left(\frac{1}{\ell h^2}\right),$$ and thus the order of mean squared error of $\widehat{\nabla}f(x)$ satisfies
	$$\mathbb{E}\|\widehat{\nabla}f(x)-\nabla f(x)\|^{2}=O(h^4)+O\left(\frac{1}{\ell h^2}\right).$$
\end{theorem}
\begin{proof}
By Taylor's expansion, we have
\begin{eqnarray*}
&&\frac{f(x+h\Delta_{j})-f(x-h\Delta_{j})}{2h\delta_{i,j}}=\sum_{k}\nabla_{k}f(x)\frac{\delta_{k,j}}{\delta_{i,j}}\\
&&+\frac{h^2}{6}\sum_{k_{1},k_{2},k_{3}}\nabla_{k_{1}k_{2}k_{3}}^{3}f(x)\frac{\delta_{k_{1},j}\delta_{k_{2},j}\delta_{k_{3},j}}{\delta_{i,j}}+o(h^2)
\end{eqnarray*}
for any $i=1,\ldots,d$, $j=1,\ldots,\ell$. Thus as $h\to0$, the expected bias of the estimator $\widehat{\nabla}_{i}f(x)$ is
\begin{eqnarray*}
	&&\mathbb{E}[\widehat{\nabla}_{i}f(x)-\nabla_{i} f(x)]\\
	&=&\mathbb{E}\left[\frac{\widehat{f}(x+h\Delta_{j})-\widehat{f}(x-h\Delta_{j})}{2h\delta_{i,j}}\right]-\nabla_{i} f(x)\\
	&=&\mathbb{E}\left[\mathbb{E}\left[\frac{\widehat{f}(x+h\Delta_{j})-\widehat{f}(x-h\Delta_{j})}{2h\delta_{i,j}}\Bigg|\Delta_{j}\right]\right]-\nabla_{i} f(x)\\
	&=&\mathbb{E}\left[\frac{f(x+h\Delta_{j})-f(x-h\Delta_{j})}{2h\delta_{i,j}}\right]-\nabla_{i} f(x)\\
	&=&\frac{h^2}{6}\left(\nabla_{iii}^{3}f(x)+3\sum_{k\neq i}\nabla_{ikk}^{3}f(x)\right)+o(h^2).
\end{eqnarray*}
On the other hand, the variance of the estimator $\widehat{\nabla}f(x)$ is
\begin{eqnarray*}
	&&\mathbb{E}\|\widehat{\nabla}f(x)-\mathbb{E}[\widehat{\nabla}f(x)]\|^{2}\\
	&=&\frac{1}{\ell}\sum_{i=1}^{d}\text{Var}\left[\frac{\widehat{f}(x+h\Delta_{j})-\widehat{f}(x-h\Delta_{j})}{2h\delta_{i,j}}\right]\\
	&=&\frac{1}{\ell}\sum_{i=1}^{d}\left(\text{Var}\left[\mathbb{E}\left[\frac{\widehat{f}(x+h\Delta_{j})-\widehat{f}(x-h\Delta_{j})}{2h\delta_{i,j}}\Bigg|\Delta_{j}\right]\right]\right.\\
	&&\left.+\mathbb{E}\left[\text{Var}\left[\frac{\widehat{f}(x+h\Delta_{j})-\widehat{f}(x-h\Delta_{j})}{2h\delta_{i,j}}\Bigg|\Delta_{j}\right]\right]\right)\\
	&=&\frac{1}{\ell}\sum_{i=1}^{d}\left(\text{Var}\left[\frac{f(x+h\Delta_{j})-f(x-h\Delta_{j})}{2h\delta_{i,j}}\right]\right.\\
	&&\left.+\mathbb{E}\left[\frac{\sigma_{f}^{2}(x+h\Delta_{j})+\sigma_{f}^{2}(x-h\Delta_{j})}{4h^2}\right]\right).
\end{eqnarray*}
Since $f(x)$ is third-order continuously differentiable in compact set $\prod_{i}X_{i}$,  the first and third order derivatives of $f(x)$ at $x\in\prod_{i}X_{i}$ are bounded.
Then we have
\begin{small}
\begin{eqnarray*}
	&&\text{Var}\left[\frac{f(x+h\Delta_{j})-f(x-h\Delta_{j})}{2h\delta_{i,j}}\right]\\
	&\leq&\mathbb{E}\left[\frac{f(x+h\Delta_{j})-f(x-h\Delta_{j})}{2h\delta_{i,j}}\right]^2\\
	&\leq&\left[\sum_{k}\Big|\nabla_{k}f(x)\Big|+\frac{h^2}{6}\sum_{k_{1},k_{2},k_{3}}\Big|\nabla_{k_{1}k_{2}k_{3}}^{3}f(x)\Big|+o(h^2)\right]^2\\
	&<&+\infty.
\end{eqnarray*}
\end{small}By bias-variance decomposition, we have
\begin{eqnarray*}
	&&\mathbb{E}\|\widehat{\nabla}f(x)-\nabla f(x)\|^{2}\\
	&=&\|\mathbb{E}[\widehat{\nabla}f(x)-\nabla f(x)]\|^{2}+\mathbb{E}\|\widehat{\nabla}f(x)-\mathbb{E}[\widehat{\nabla}f(x)]\|^{2}.
\end{eqnarray*}
Therefore as $h\to0$ and $\ell\to\infty$, we have
\begin{small}
\begin{eqnarray*}
	&&\mathbb{E}\|\widehat{\nabla}f(x)-\nabla f(x)]\|^{2}\\
	&=&\sum_{i=1}^{d}\left(\left[\frac{h^2}{6}\left(\nabla_{iii}^{3}f(x)+3\sum_{k\neq i}\nabla_{ikk}^{3}f(x)\right)\right]^2+o(h^4)\right.\\
	&&\left.+\frac{1}{\ell}\mathbb{E}\left[\frac{\sigma_{f}^{2}(x+h\Delta_{j})+\sigma_{f}^{2}(x-h\Delta_{j})}{4h^2}\right]+o\left(\frac{1}{\ell h^2}\right)\right).
\end{eqnarray*}
\end{small}Since $\sigma_{f}^{2}(x)$ for any $x\in\prod_{i}X_{i}$ is bounded, as $h\to0$ and $\ell\to\infty$,
\begin{eqnarray*}
	\mathbb{E}\|\widehat{\nabla}f(x)-\nabla f(x)]\|^{2}=O(h^4)+O(\frac{1}{\ell h^2}).
\end{eqnarray*}
\end{proof}

By simultaneous perturbation, each player $i\in\mathcal{N}$ can estimate the partial derivative of its cost function with respect to its own decision variables by $\widehat{\nabla}_{x_{i}}f_{i}(x_{i},x_{-i})$. Then we define $$\widehat{\phi}(x)\triangleq(\widehat{\nabla}_{x_{1}}f_{1}(x_{1},x_{-1})^\top,\ldots,\widehat{\nabla}_{x_{N}}f_{N}(x_{N},x_{-N})^\top)^{\top}.$$
Since each player cannot observe other players’ perturbations, the partial derivative of its cost function with respect to its rivals' decision variables still cannot be obtained.

\begin{remark}
	The scaling parameter $h$ must be chosen with consideration of controlling the bias and variance that contributes to the mean squared error. As the scaling parameter $h$ increases, bias increases while variance decreases (and vice versa). To minimize the mean squared error, the optimal $h$ turns out to be of order $\ell^{-1/6}$. However, we will show that the optimal $h$ for the mean squared error does not achieve an optimal convergence rate for seeking the Nash equilibrium.
\end{remark}

\begin{remark}~\label{rm}
In~\cite{bravo2018bandit},	they use a single-shot estimate, i.e., $\widehat{\nabla}_{x_{i}}f_{i}(x_{i},x_{-i})=\frac{m_{i}}{h}\widehat{f}_{i}(\tilde{x}_{i},\tilde{x}_{-i})z_{i}$, where random vector $z_{i}$ is uniformly drawn from the unit sphere $\mathbb{S}^{m_{i}}$ and $\tilde{x}_{i}=x_{i}+hz_{i}$. As $h\to 0$, the bias of this estimate diminishes at a rate $O(h)$, whereas the variance of this estimate grows at a rate $O(1/h^{2})$. Therefore, when this estimate and simultaneous perturbation method have the same order of bias, the variance of this estimate is the square of that of simultaneous perturbation method. We will show how errors in gradient estimation affect the convergence rate in searching for the Nash equilibrium.
\end{remark}

\subsection{Mirror Descent}
\begin{figure*}[htbp]
	\centering
	\includegraphics[scale=1.1,clip,keepaspectratio]{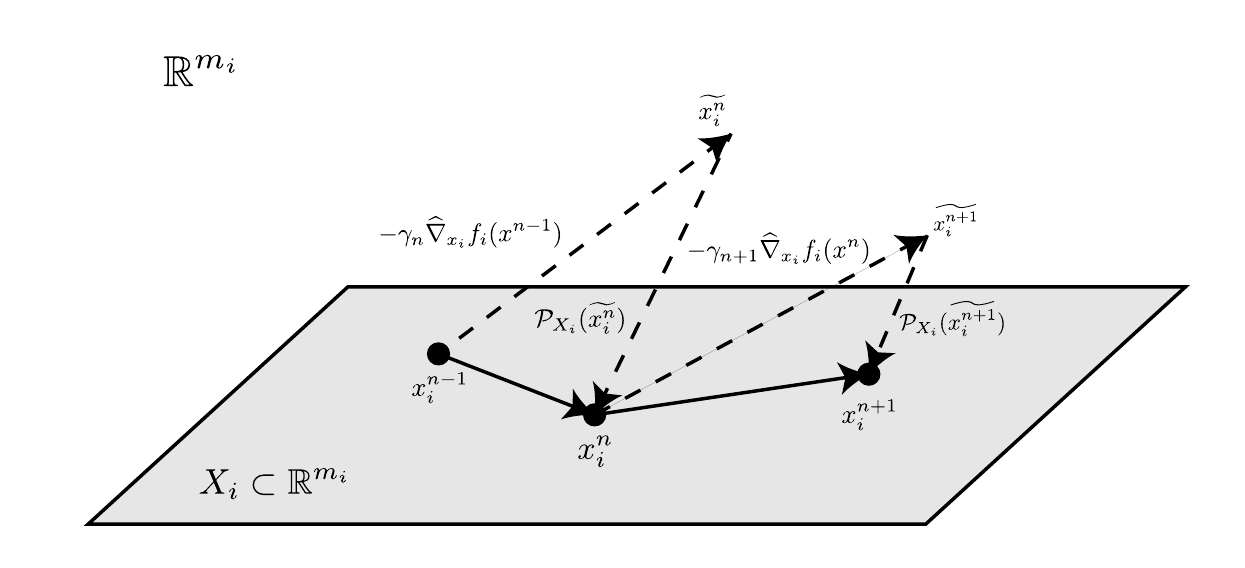}
	\caption{Schematic Illustration of Mirror Descent.}
	\label{fig:MD}
\end{figure*}
We put optimization problem~\eqref{game} into a more general form:
\begin{equation*}
	\forall i \in \mathcal{N}:\ \min_{x_i \in \mathbb{R}^{m_i}} g_i(x_i, x_{-i}) = f_i(x_i, x_{-i}) + P_i(x_i),
\end{equation*}
where $f_i$ is a convex and smooth function and 
\begin{equation*}
	P_i(x_i) = \left \{
	\begin{aligned}
		&0, \quad &x_i \in X_i, \\
		&+\infty, \quad &x_i \notin X_i,
	\end{aligned}
	\right.
\end{equation*}
is a convex but nonsmooth function, induced by the strategy set $X_i$ (convex set constraint).

For such a composite nonsmooth convex minimization problem, we 
propose a mirror descent method with a stochastic gradient estimate to seek the Nash equilibrium. Specifically, suppose that each player $i\in\mathcal{N}$ is endowed with a continuous $\sigma_{i}$-strongly convex function $h_{i}:X_{i}\to\mathbb{R}$, i.e.,
$$h_{i}(x_{i}')\geq h_{i}(x_{i})+\langle\nabla h_{i}(x_{i}),x_{i}'-x_{i}\rangle+\frac{\sigma_{i}}{2}\|x_{i}'-x_{i}\|^{2}$$
for all $x_{i},x_{i}'\in X_{i}$. The Bregman divergence $D_{i}(x_{i},x_{i}')$ of the function $h_{i}$ is defined as the difference between $h_{i}(x_{i})$ and the best linear approximation of $h_{i}(x_{i})$ from $x_{i}'$, i.e.,
$$D_{i}(x_{i},x_{i}')\triangleq h_{i}(x_{i})-h_{i}(x_{i}')-\langle\nabla h_{i}(x_{i}'),x_{i}-x_{i}'\rangle.$$
Then mirror descent scheme starts with an arbitrary $x^0\in \prod_{i}X_{i}$, and proceeds in stage $n$ by the recursion
\begin{eqnarray}
	&x_{i}^{n}=\arg\min&\{\langle -\gamma_{n}\widehat{\nabla}_{x_{i}}f_{i}(x_{i}^{n-1},x_{-i}^{n-1}),x_{i}^{n-1}-x_{i}\rangle\nonumber\\
	&&+D_{i}(x_{i},x_{i}^{n-1}):x_{i}\in X_{i}\},\label{eq:MD}
\end{eqnarray}
where $\gamma_{n}>0$ is a nonincreasing step-size sequence.

In the rest of this paper, we use Euclidean projections as a standard example, i.e., $h(x)=\frac{1}{2}\|x\|^{2}$ and $D(x,x')=\frac{1}{2}\|x-x'\|^{2}$. In this case, our method can be considered as a projected gradient method~\cite{nemirovskij1983problem}
\begin{eqnarray*}
	x_i^n &=& \mathcal{P}_{X_i}(\widetilde{x_{i}^{n}})\\
	 &\triangleq& \arg\min\{\| x_i - \widetilde{x_{i}^{n}} \|^{2}: x_i \in X_i \},
\end{eqnarray*}
where $\widetilde{x_{i}^{n}}\triangleq x_i^{n-1} - \gamma_n \widehat{\nabla}_{x_i} f_i(x^{n-1})$. Intuitively, the main idea of mirror descent method is as follows: at each stage $n$, each player $i$ takes a step from its current strategy $x_{i}^{n}$ along a stochastic gradient estimate $\widehat{\nabla}_{x_i} f_i(x^{n})$, and then mirrors the output $\widetilde{x_{i}^{n+1}}$ back to the feasible space $X_{i}$. See Figure~\ref{fig:MD} for a schematic illustration.

Note that we have $D_{i}(x_{i},x_{i}')=0$ if and only if $x_{i}=x_{i}'$. Thus the convergence of a sequence $\{x_{i}^{n}\}$ to $x_{i}^*$ can be checked by showing that $D_{i}(x_{i}^{n},x_{i}^*)\to0$. The converse that $D_{i}(x_{i}^{n},x_{i}^*)\to0$ when $x_{i}^{n}\to x_{i}^*$ is referred as to ``Bregman reciprocity"~\cite{chen1993convergence}, which is trivially
satisfied by the Euclidean norm.

\subsection{Stochastic Distributed Learning Algorithm}
Based on simultaneous perturbation and mirror descent methods, we propose a stochastic distributed learning algorithm. In period $n$, players perturb their strategies $x^{n-1}\in \prod_{i}X_{i}$ by $\pm\Delta_{j},j=1,\ldots,\ell$, receive observations of their cost functions in $2\ell$ games, and compute the gradient estimate $\widehat{\phi}(x^{n-1})$. At the end of period $n$, each player updates its strategy $x_{i}^{n}$ by mirror descent method. The procedures are summarized in Algorithm~\ref{alg:SDL}.
\begin{algorithm}[h]
	\caption{Stochastic Distributed Learning (SDL)}
	\label{alg:SDL}
	\begin{algorithmic}
		\STATE {\bfseries Input:} $n=1$, $x_{i}^{0}\in X_{i}$ for each $i\in\mathcal{N}$, $(\ell_{n})_{n=1}^{\infty}$, $(h_{n})_{n=1}^{\infty}$, and $(\gamma_{n})_{n=1}^{\infty}$.
		\REPEAT
		\FOR{$j=1$ {\bfseries to} $\ell_{n}$}
		\STATE Generate a perturbation vector $\Delta_{j}$, and observe the stochastic costs $\widehat{f}_{i}(x^{n-1}+h_{n}\Delta_{j})$ and $\widehat{f}_{i}(x^{n-1}-h_{n}\Delta_{j})$, $i\in\mathcal{N}$.		
		\ENDFOR
		\STATE Compute the gradient estimate $\widehat{\phi}(x^{n-1})$ by~\eqref{eq:SP}.
		\STATE Each player $i\in\mathcal{N}$ updates its equilibrium strategy by~\eqref{eq:MD}.
		\UNTIL{stopping criterion is satisfied}
	\end{algorithmic}
\end{algorithm}

\begin{remark}
	In Algorithm~\ref{alg:SDL}, players perform their own computation in a distributed manner. When estimating the partial derivative, each player only uses its own strategy perturbations and cost function evaluations. When applying mirror descent method, the information required by each player includes its current strategy and the partial derivative of its cost function with respect to its own decision variables. It is worth noting that Algorithm~\ref{alg:SDL} can learn the Nash equilibrium without local information exchange.
\end{remark}

The following theorem shows that our proposed learing algorithm converges to the Nash equilibrium in strongly monotone games.
\begin{theorem}\label{thm:NE}
	Suppose Assumptions~\ref{Assump:1} and~\ref{Assump:2} hold. Let $x^{*}$ be the unique Nash equilibrium of a $\beta$-strongly monotone game, and $x^{n}$ be generated by Algorithm~\ref{alg:SDL} with parameters $\gamma_{n}=\gamma/n>0$, $\ell_{n}=\ell_{0}n^{p}>0$, and $h_{n}=h_{0}n^{-(p+1)/4}>0$, where $\gamma$, $\ell_{0}$, and $h_{0}$ are positive constants. Then,\\
	(a) if $0\leq p\leq 1$ and $\gamma>1/\beta$, we have
	$$\mathbb{E}[\|x^{n}-x^{*}\|^2]=O(n^{-(p+1)/2}).$$
	(b) if $p>1$ and $\gamma>1/\beta$, we have
	$$\mathbb{E}[\|x^{n}-x^{*}\|^2]=O(n^{-1}).$$
\end{theorem}
\begin{proof}
We first provide the following upper and lower bounds of the Bregman divergence.
\begin{proposition}\label{thm:BD}
	\cite{bravo2018bandit}. Let $h_{i}$ be a $\sigma_{i}$-strongly convex function on $X_{i}$. For $x_{i},p_{i}\in X_{i}$ and $x_{i}^{+}=P_{i,x_{i}}(y_{i})=\arg\min\{\langle y_{i},x_{i}-x_{i}'\rangle+D_{i}(x_{i}',x_{i}):x_{i}'\in X_{i}\}$, we have
	\begin{eqnarray*}
		D_{i}(p_{i},x_{i})&\geq&\frac{\sigma_{i}}{2}\|p_{i}-x_{i}\|^{2},\\
		D_{i}(p_{i},x_{i}^{+})&\leq&D_{i}(p_{i},x_{i})-D_{i}(x_{i}^{+},x_{i})+\langle y_{i},x_{i}^{+}-p_{i}\rangle\\
		&\leq&D_{i}(p_{i},x_{i})+\langle y_{i},x_{i}-p_{i}\rangle+\frac{1}{2\sigma_{i}}\|y_{i}\|^{2}.
	\end{eqnarray*}
\end{proposition}

Applying Proposition~\ref{thm:BD} to Euclidean projections $h_{i}(x_{i})=\frac{1}{2}\|x_{i}\|^2$, we have
$$D_{i}(x_{i}^{*},x_{i}^{n})=\frac{1}{2}\|x_{i}^{n}-x_{i}^{*}\|^2$$
and
\begin{small}
\begin{eqnarray*}
	&&D_{i}(x_{i}^*,x_{i}^{n+1})\\
	&\leq&D_{i}(x_{i}^*,x_{i}^{n})+\langle- \gamma_{n+1}\widehat{\nabla}_{x_{i}}f_{i}(x_{i}^{n},x_{-i}^{n}),x_{i}^{n}-x_{i}^*\rangle\\
	&&+\frac{1}{2\sigma_{i}}\|\gamma_{n+1}\widehat{\nabla}_{x_{i}}f_{i}(x_{i}^{n},x_{-i}^{n})\|^{2}\\
	&=&D_{i}(x_{i}^*,x_{i}^{n})-\gamma_{n+1}\langle\nabla_{x_{i}}f_{i}(x_{i}^{n},x_{-i}^{n}),x_{i}^{n}-x_{i}^*\rangle\\
	&&-\gamma_{n+1}\langle\widehat{\nabla}_{x_{i}}f_{i}(x_{i}^{n},x_{-i}^{n})-\mathbb{E}[\widehat{\nabla}_{x_{i}}f_{i}(x_{i}^{n},x_{-i}^{n})],x_{i}^{n}-x_{i}^*\rangle\\
	&&-\gamma_{n+1}\langle\mathbb{E}[\widehat{\nabla}_{x_{i}}f_{i}(x_{i}^{n},x_{-i}^{n})]-\nabla_{x_{i}}f_{i}(x_{i}^{n},x_{-i}^{n}),x_{i}^{n}-x_{i}^*\rangle\\
	&&+\frac{\gamma_{n+1}^{2}}{2\sigma_{i}}\|\widehat{\nabla}_{x_{i}}f_{i}(x_{i}^{n},x_{-i}^{n})\|^{2}.
\end{eqnarray*}
\end{small}Let $D_{n}=D(x^{*},x^{n})=\sum_{i=1}^{N}D_{i}(x_{i}^{*},x_{i}^{n})$ for a Nash equilibrium $x^*$. Then we have
$$D_{n}=\frac{1}{2}\|x^{n}-x^{*}\|^{2}$$
and
\begin{eqnarray}
	D_{n+1}&\leq&D_{n}-\gamma_{n+1}\langle\phi(x^{n}),x^{n}-x^*\rangle\nonumber\\
	&&-\gamma_{n+1}\langle\widehat{\phi}(x^{n})-\mathbb{E}[\widehat{\phi}(x^{n})],x^{n}-x^*\rangle\nonumber\\
	&&-\gamma_{n+1}\langle\mathbb{E}[\widehat{\phi}(x^{n})]-\phi(x^{n}),x^{n}-x^*\rangle\nonumber\\
	&&+\frac{\gamma_{n+1}^{2}}{2\sigma}\|\widehat{\phi}(x^{n})\|^{2},\label{ineq}
\end{eqnarray}
where $\sigma=\min\{\sigma_{i},i\in\mathcal{N}\}$.

Since the game is $\beta$-strongly monotone and $x^*$ is a NE, we have
\begin{eqnarray*}
	\langle\phi(x^{n}),x^{n}-x^*\rangle&\geq&\langle\phi(x^{n})-\phi(x^*),x^{n}-x^*\rangle\\
	&\geq&\frac{\beta}{2}\|x^{n}-x^*\|^{2}\\
	&=&\beta D_{n}.
\end{eqnarray*}
Then taking expectations of both sides of~\eqref{ineq}, we have
\begin{eqnarray}
	\mathbb{E}[D_{n+1}]&\leq& (1-\beta\gamma_{n+1})\mathbb{E}[D_{n}]\nonumber\\
	&&-\gamma_{n+1}\langle\mathbb{E}[\widehat{\phi}(x^{n})]-\phi(x^{n}),x^{n}-x^*\rangle\nonumber\\
	&&+\frac{\gamma_{n+1}^{2}}{2\sigma}\mathbb{E}\|\widehat{\phi}(x^{n})\|^{2}.\label{ineq2}
\end{eqnarray}
By the proof of Theorem~\ref{thm:mse}, we have
$$\mathbb{E}[\widehat{\phi}(x^{n})]-\phi(x^{n})=O(h_{n+1}^{2})$$
and
$$\mathbb{E}\|\widehat{\phi}(x^{n})\|^{2}=O(1)+O(h_{n+1}^{4})+O\left(\frac{1}{\ell_{n+1} h_{n+1}^{2}}\right).$$

If $\ell_{n}=\ell_{0}n^{p}>0$ with $0\leq p\leq 1$ and $h_{n}=h_{0}n^{-(p+1)/4}>0$, then we have
$$\mathbb{E}[\widehat{\phi}(x^{n})]-\phi(x^{n})=O(n^{-(p+1)/2})$$
and
$$\mathbb{E}\|\widehat{\phi}(x^{n})\|^{2}=O(n^{(1-p)/2}).$$
That is, when $n$ is large enough, there exists $A>0$ and $B>0$ such that
$$|\langle\mathbb{E}[\widehat{\phi}(x^{n})]-\phi(x^{n}),x^{n}-x^*\rangle|\leq An^{-(1+p)/2}$$
and
$$\mathbb{E}\|\widehat{\phi}(x^{n})\|^{2}\leq Bn^{(1-p)/2}.$$
Therefore, given $\gamma_{n}=\gamma/n$, we have 
\begin{gather}
	\mathbb{E}[D_{n+1}]\leq\left(1-\frac{\beta\gamma}{n+1}\right)\mathbb{E}[D_{n}]+\frac{C}{n^{(3+p)/2}},\label{relation1}
\end{gather}
where $C=A\gamma+\frac{\gamma^2 B}{2\sigma}$.

If $\ell_{n}=\ell_{0}n^{p}>0$ with $p>1$ and $h_{n}=h_{0}n^{-(p+1)/4}>0$, then we have
$$\mathbb{E}[\widehat{\phi}(x^{n})]-\phi(x^{n})=O(n^{-(p+1)/2})$$
and
$$\mathbb{E}\|\widehat{\phi}(x^{n})\|^{2}=O(1).$$
That is, when $n$ is large enough, there exists $A'>0$ and $B'>0$ such that
$$|\langle\mathbb{E}[\widehat{\phi}(x^{n})]-\phi(x^{n}),x^{n}-x^*\rangle|\leq A'n^{-(1+p)/2}$$
and
$$\mathbb{E}\|\widehat{\phi}(x^{n})\|^{2}\leq B'.$$
Therefore, given $\gamma_{n}=\gamma/n$, we have 
\begin{gather}
	\mathbb{E}[D_{n+1}]\leq\left(1-\frac{\beta\gamma}{n+1}\right)\mathbb{E}[D_{n}]+\frac{C'}{n^2},\label{relation2}
\end{gather}
where $C'>A'\gamma n^{(1-p)/2}+\frac{\gamma^2 B'}{2\sigma}$.

We consider a more general form of~\eqref{relation1} and~\eqref{relation2}, and use the following lemma to conclude Theorem~\ref{thm:NE}.
\begin{lemma}\label{lemma:q}
	\cite{chung1954stochastic}. $\mathbb{E}[D_{n}], n=1,2,\ldots,$ is a non-negative sequence such that
	\begin{gather}
		\mathbb{E}[D_{n+1}]\leq \left(1-\frac{P}{n+1}\right)\mathbb{E}[D_{n}]+\frac{Q}{n^{1+q}},\label{eq.general}
	\end{gather}
	where $P,Q,q>0$ and $P>q$. Then we have
	$$\mathbb{E}[D_{n}]=O(n^{-q}).$$
\end{lemma}

Thus, if $\ell_{n}=\ell_{0}n^{p}>0$ with $0\leq p\leq 1$ and $h_{n}=h_{0}n^{-(p+1)/4}>0$, applying Lemma~\ref{lemma:q} with $\beta\gamma>1$ and $q=(p+1)/2$, we have
$$\mathbb{E}[\|x^{n}-x^{*}\|^2]=2\mathbb{E}[D_{n}]=O(n^{-(p+1)/2}).$$
If $\ell_{n}=\ell_{0}n^{p}>0$ with $p>1$ and $h_{n}=h_{0}n^{-(p+1)/4}>0$, applying Lemma~\ref{lemma:q} with $\beta\gamma>1$ and $q=1$, we have
$$\mathbb{E}[\|x^{n}-x^{*}\|^2]=2\mathbb{E}[D_{n}]=O(n^{-1}).$$
\end{proof}

The proof of Theorem~\ref{thm:NE} shows how errors in gradient estimation affect the convergence rate in searching for the Nash equilibrium. As shown in equation~\eqref{ineq2}, the bias $\mathbb{E}[\widehat{\phi}(x^{n})]-\phi(x^{n})$ and the second moment $\mathbb{E}\|\widehat{\phi}(x^{n})\|^{2}$ control the convergence rate of $\mathbb{E}[\|x^{n}-x^{*}\|^{2}]$. As the scaling parameter $h$ increases, bias increases while variance decreases (and vice versa). To maximize the convergence rate, the scaling parameter $h$ is determined by optimally balancing the magnitudes of the two error sources. Note that the optimal $h$ shown in Theorem~\ref{thm:NE} is different from that for minimizing the mean squared error of gradient estimate.

Using the single-shot gradient estimate in~\cite{bravo2018bandit}, a Nash equilibrium learning scheme leads to an $O(n^{-1/3})$ convergence rate. In our proposed SDL, the number of games required to estimate gradient is $2\ell$, and $p$ is the order of $\ell$ with respect to $n$. When $p$ is set to 0, SDL uses a constant number of observations to estimate the gradient, and achieves an $O(n^{-1/2})$ convergence rate as shown in Theorem~\ref{thm:NE}. The improvement in convergence rate is due to that SDL has finer control on the bias/variance of gradient estimate. As discussed in Remark~\ref{rm}, when the orders of bias are set the same, the variance of gradient estimate in~\cite{bravo2018bandit} is much higher than that in SDL.

The setting of $p>0$ represents more games can be performed around the current strategy. More observations of the game can reduce the bias and variance of gradient estimate. When $0<p\leq1$, reducing estimation error can improve the convergence rate. Such improvement potential also indicates that stochastic gradient estimate has not been accurate enough for the current mirror descent method. When $p>1$, reducing estimation error will no longer affect the convergence rate. Our proposed SDL at most achieves an $O(n^{-1})$ convergence rate, which is also the convergence rate under the assumption of unbiased, finite mean-square gradient estimate~\cite{bravo2018bandit}.

\section{Numerical Simulations}\label{sec:NS}
In this section, we conduct numerical simulations to validate the performance of the proposed SDL. The learning algorithm in~\cite{bravo2018bandit} which uses the single-shot gradient estimate is also tested as comparison.

We consider the Cournot competition problem in Example~\ref{example}. The capacity constraint of firm $i\in\mathcal{N}$ is $X_{i}\triangleq\{x_{i}\geq0:\mathbf{1}^\top x_{i}=1\}$. Then update~\eqref{eq:MD} becomes
\begin{eqnarray*}
x_{i}^{n}&=&x_{i}^{n-1}+\frac{1}{m_{i}}\mathbf{1}\mathbf{1}^\top\gamma_{n}\widehat{\nabla}_{x_{i}}f_{i}(x_{i}^{n-1},x_{-i}^{n-1})\\
&&-\gamma_{n}\widehat{\nabla}_{x_{i}}f_{i}(x_{i}^{n-1},x_{-i}^{n-1}).
\end{eqnarray*}
We set $N=20$, $m_{i}=m=5, i\in\mathcal{N}$. For each firm $i\in\mathcal{N}$, each $c_{i,j}$ is randomly generated from the uniform distribution $U[3,4]$. For each market $j\in\mathcal{M}$, $a_{j}$ and $b_{j}$ are randomly generated from uniform distributions $U[4,5]$ and $U[0.5,0.55]$, respectively. The random variables $\zeta_{j}$ and $\eta_{i,j}$ are drawn from uniform distributions $U[-a_{j}/8,a_{j}/8]$ and $U[-c_{i,j}/8,c_{i,j}/8]$, respectively. 

We implement all algorithms with $\gamma_{n}=\frac{1}{2n}$, and test the proposed SDL with different settings of $p$. Figure~\ref{fig:result} shows that all the
squared errors $\|x^{n}-x^{*}\|^2$ of three algorithms asymptotically decrease to zero, which implies three algorithms converge to the Nash equilibrium. Our proposed SDL has a faster convergence rate than the learning algorithm using single-shot gradient estimates. In addition, as $p$ increases from 0 to 1, the convergence rate of SDL is further improved. Such improvement could be attributed to that more observations of the game reduce errors in gradient estiamtion.
\begin{figure}[htbp]
	\centering
	\includegraphics[scale=0.24,clip,keepaspectratio]{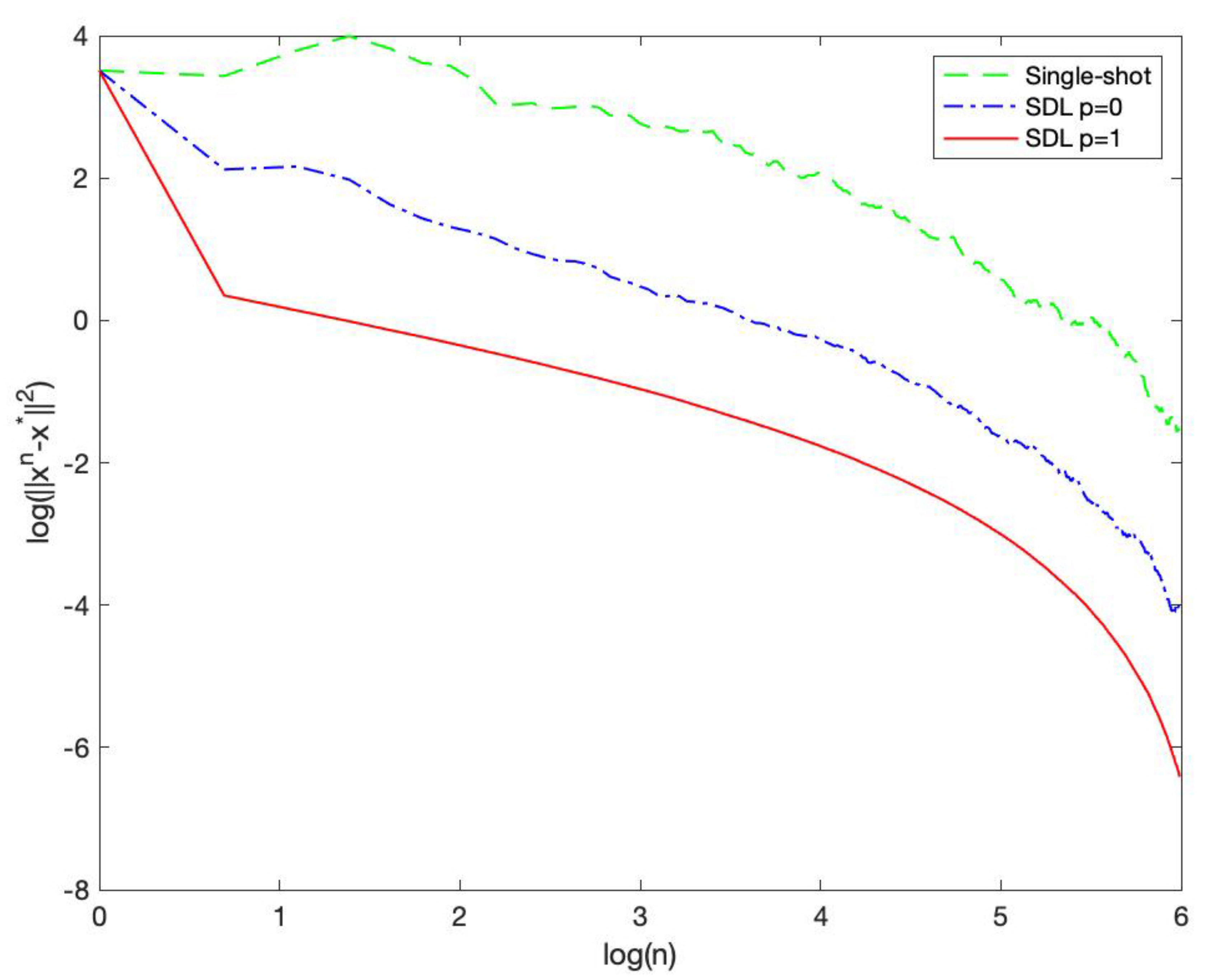}
	\caption{Convergence to Nash equilibrium.}
	\label{fig:result}
\end{figure}

\section{Conclusions}\label{sec:con}
This paper proposes an efficient distributed learning algorithm for stochastic non-cooperative games without information exchange. The asymptotic analysis shows the superiority of simultaneous perturbation method. For
a class of strictly monotone games, our proposed algorithm is shown to converge in mean square to the Nash equilibrium, which is faster than the existing method. In future work, we could investigate a fully decentralized setting where the players' updates need not be synchronous.

\bibliographystyle{unsrt}
\bibliography{example_paper}

%

\end{document}